\documentclass[11pt]{article}
\usepackage[utf8]{inputenc}

\RequirePackage[l2tabu,orthodox]{nag}

\usepackage[ruled,vlined,linesnumbered]{algorithm2e}
\usepackage{amsmath}
\usepackage{amssymb}
\usepackage{amsthm}
\usepackage{array}
\usepackage{authblk}
\usepackage[american]{babel}
\usepackage{booktabs}
\usepackage[nobreak]{cite}
\usepackage{fixltx2e}
\usepackage{fullpage}
\usepackage{graphicx}
\usepackage[bb=boondox]{mathalfa}
\usepackage{microtype}
\usepackage{multirow}
\usepackage[autolanguage]{numprint}
\usepackage{paralist}
\usepackage{siunitx}
\usepackage{standalone}
\usepackage{tikz}
\usetikzlibrary{positioning,calc}
\usepackage{textcomp}
\usepackage{todonotes}
\usepackage{url}
\usepackage{wrapfig}
\usepackage{xspace}

\let\OLDtodo=\todo
\renewcommand{\todo}[1]{\OLDtodo[inline,,caption={}]{#1}}

\newtheorem{definition}{Definition}
\newtheorem{theorem}[definition]{Theorem}

\newtheorem{lemma}[definition]{Lemma}

\newcommand{\omu}{\overline{\mu}}
\newcommand{\onu}{\overline{\nu}}
\newcommand{\oalpha}{\overline{\alpha}}
\newcommand{\obeta}{\overline{\beta}}
\newcommand{\Om}{\Omega_{n,m}^N(\omu,\onu)}
\newcommand{\Am}{A_\mathrm{min}}
\newcommand{\Bm}{B_\mathrm{max}}
\newcommand{\Prob}{\mathbb{P}}
\newcommand{\OMEGA}{\Om[\Am,\Bm]}

\begin{document}

\title{A Markov Chain Monte Carlo Approach to Cost Matrix Generation
  for Scheduling Performance Evaluation}

%% \author{\IEEmail: \verb+{louis-claude.canon | mohamad.el_sayah | pierre-cyrille.heam@univ-fcomte.fr}+EEauthorblockN{
%% Louis-Claude Canon\IEEEauthorrefmark{1}\IEEEauthorrefmark{2},
%% Mohamad El Sayah\IEEEauthorrefmark{2},
%% Pierre-Cyrille Héam\IEEEauthorrefmark{2}
%% }
%% \IEEEauthorblockA{\IEEEauthorrefmark{1}LIP, \'Ecole Normale Sup\'erieure de Lyon, CNRS \& Inria, France}
%% \IEEEauthorblockA{\IEEEauthorrefmark{2}FEMTO-ST, Université de Bourgogne Franche-Comté, France}
%% Email: \url{{louis-claude.canon | mohamad.el_sayah | pierre-cyrille.heam@univ-fcomte.fr}}
%% }

\author[1,2]{Louis-Claude Canon} 
\author[2]{Mohamad El Sayah} 
\author[2]{Pierre-Cyrille Héam}
\affil[1]{\footnotesize \'Ecole Normale Sup\'erieure de Lyon, CNRS \& Inria, France}
\affil[2]{ FEMTO-ST Institute, CNRS, Univ. Bourgogne Franche-Comt\'e,
  France, }

\maketitle

\begin{abstract}
  In high performance computing, scheduling of tasks and allocation to
  machines is very critical especially when we are dealing with
  heterogeneous execution costs.
  Simulations can be performed with a large variety of environments
  and application models.
  However, this technique is sensitive to bias when it relies on
  random instances with an uncontrolled distribution.
  We use methods from the literature to provide formal guarantee on
  the distribution of the instance.
  In particular, it is desirable to ensure a uniform distribution
  among the instances with a given task and machine heterogeneity.
  In this article, we propose a method that generates instances (cost
  matrices) with a known distribution where tasks are scheduled on
  machines with heterogeneous execution costs.
\end{abstract}

%%%%%%%%%%%%%%%%%%%%%%%%%%%%%%%%%%%%%%%%%%%%%%%%%%%%%%%%%%%%%%%%%%%%%%%%%%
\section{Introduction}

Empirical assessment is critical to determine the best scheduling
heuristics on any parallel platform.
However, the performance of any heuristic may be specific to a given
parallel computer.
In addition to experimentation on real platforms, simulation is an
effective tool to quantify the quality of scheduling heuristics.
Even though simulations provide weaker evidence, they can be performed
with a large variety of environments and application models, resulting
in broader conclusions.
However, this technique is sensitive to bias when it relies on random
instances with an uncontrolled or irrelevant distribution.
For instance, in uniformly distributed random graphs, the probability
that the diameter is 2 tends exponentially to 1 as the size of the
graph tends to infinity\cite{DBLP:journals/jsyml/Fagin76}.
Even though such instances may be sometimes of interest, they prove
useless in most practical contexts.
We propose a method that generates instances with a known distribution
for a set of classical problems where tasks must be scheduled on
machines (or processors) with heterogeneous execution costs.
This is critical to the empirical validation of many new heuristics
like BalSuff\cite{cp15a} for the problem $R||C_{max}$ and
PEFT\cite{arabnejad2014list} for $R|\textrm{prec}|C_{\max}$ in
Graham's notation\cite{graham79a}.

In this context, an instance consists of a $n\times m$ cost matrices,
$M$, where the element of row $i$ and column $j$, $M(i,j)$, represents
the execution cost of task $i$ on machine $j$.
Like the diameter for graphs, multiple criteria characterize cost
matrices.
First, the heterogeneity can be determined globally with the variance
of all costs, but also relatively to the rows or columns.
For instance, the dispersion of the means on each row, which
corresponds to the varying costs for each task, impacts the
performance of some scheduling
heuristics\cite{canon2017heterogeneity}.
The correlations between the rows and columns also play an important
role as it corresponds to the machines being either related or
specialized, with some affinity between the tasks and the
machines\cite{canon2017controlling}.

Among existing methods, the shuffling one\cite{canon2017heterogeneity}
starts by an initial matrix in which rows are proportional to each
other (leading to large row and column correlations).
Then, it proceeds to mix the values in the matrix such as to keep the
same sum on each row and column.
This ensures that the row and column heterogeneity remains stable,
while the correlation decreases.
However, this approach is heuristic and provides no formal guarantee
on the distribution of the instances.
In addition, when the number of shuffles increases, the cost CV
increases, which leads to non-interpretable results (see
Figure~\ref{fig:shuffling}).

\begin{figure}
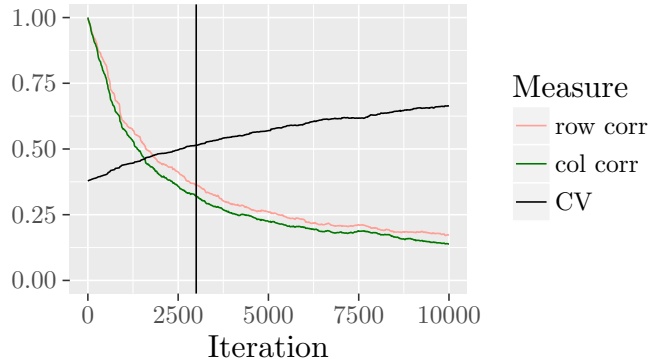

  \centering
  \includestandalone{1_limit_shuffling}
  \caption{Cost Coefficient-of-Variation (ratio of standard deviation
    to mean) and mean row and column correlations at each iteration
    of the shuffling method\cite{canon2017heterogeneity} when
    generating a $100\times30$ cost matrix.
    The shuffling method arbitrarily stops after 3\,000 iterations
    (represented by the black vertical line).}
  \label{fig:shuffling}
\end{figure}

While other methods exist, some of them with stronger formal
guarantees, it remains an open problem to ensure a uniform
distribution among the instances that have a given task and machine
heterogeneity.
Our contribution is to control the row and column heterogeneity, while
limiting the overall variance and ensuring a uniform
distribution among the set of possible instances.
The approach is based on a Markov Chain Monte Carlo process and relies
on contingency tables\footnote{A contingency table is a matrix with
the sum of each row (resp.\ column) displayed in an additional total
row (resp.\ column).
They are usually used to show the distribution of two variables.}.
More precisely, the proposed random generation process is based on two
steps. For a given $n$ (number of tasks), $m$ (number of machines) and $N$
(sum of the cost of the tasks):
\begin{enumerate}
\item Randomly generate the average cost of each task and the average speed
  of each machine. This random generation is performed uniformly using
  classical recursive algorithms\cite{DBLP:journals/tcs/FlajoletZC94}.
  In order to control the heterogeneity,
  we show how to restrict this uniform random generation to interesting
  classes of vectors. This step is described in
  Section~\ref{section:vectors}.
\item Next, the cost matrices can be generated using a classical MCMC
  approach: from an initial matrix, a random walk in the graph of
  contingency tables is performed. It is known (see for instance\cite{LevinPeresWilmer2006})
  that if the Markov Chain associated with this walk is ergodic and symmetric,
  then the unique stationary distribution exists and is uniform. Walking
  enough steps in the graph leads to any state with the same probability.
  Section~\ref{section:MC} provides several symmetric and ergodic Markov
  Chains for this problem. The main contribution of this section is to
  extend known results for contingency tables to contingency tables with
  min/max constraints.
\end{enumerate}

In order to evaluate the mixing time of the proposed Markov Chains
(the mixing time is the number of steps to walk in order to be close to
the uniform distribution), we propose practical and statistical estimations
in Section~\ref{MC:convergence}. Note that obtaining theoretical bound on
mixing time is a very hard theoretical problem, still open in the general
case of unconstrained contingency tables.
In Section~\ref{sec:const-effect-cost}, we used our random generation process
to evaluate scheduling algorithms.
The algorithms are implemented in R and Python and the related code,
data and analysis are available in\cite{figshare}.

%%%%%%%%%%%%%%%%%%%%%%%%%%%%%%%%%%%%%%%%%%%%%%%%%%%%%%%%%%%%%%%%%%%%%%%%%%
\section{Related Work}\label{section:related}

Two main methods have been used in the literature: RB (range-based)
and CVB (Coefficient-of-Variation-Based)\cite{ali2000a,ali2000b}.
Both methods follow the same principle: $n$ vectors of $m$ values are
first generated using a uniform distribution for RB and a gamma
distribution for CVB; then, each row is multiplied by a random value
using the same distribution for each method.
A third optional step consists in sorting each row in a submatrix,
which increases the correlation of the cost matrix.
However, these methods are difficult to use when generating a matrix
with given heterogeneity and low
correlation\cite{canon2017controlling,canon2017heterogeneity}.

More recently, two additional methods have been proposed for a better
control of the heterogeneity: SB (shuffling-based) and NB
(noise-based)\cite{canon2017heterogeneity}.
In the first step of SB, one column of size $n$ and one row of size
$m$ are generated using a gamma distribution.
These two vectors are then multiplied to obtain a $n\times m$ cost
matrix with a strong correlation.
To reduce it, values are shuffled without changing the sum on any row
or column as it is done is Section~\ref{section:MC}: selecting four
elements on two distinct rows and columns (a submatrix of size
$2\times2$); and, removing/adding the maximum quantity to two elements
on the same diagonal while adding/removing the same quantity to the
last two elements on the other diagonal.
While NB shares the same first step, it introduces randomness in the
matrix by multiplying each element by a random variable with expected
value one instead of shuffling the elements.
When the size of the matrix is large, SB and NB provide some control
on the heterogeneity but the distribution of the generated instances
is unknown.

Finally, CNB (correlation noise-based) and CB (combination-based) have
been proposed to control the correlation\cite{canon2017controlling}.
CNB is a direct variation of CB to specify the correlation more
easily.
CB combines correlated matrices with an uncorrelated one to obtain the
desired correlation.
As for SB and NB, both methods have asymptotic guarantees when the
size of the matrix tends to infinity, but no guarantee on how
instances are distributed.

The present work relies on contingency tables/matrices, which are
important data structures used in statistics
for displaying the multivariate frequency distribution of variables,
introduced in 1904 by K. Pearson\cite{pearson}. The MCMC approach is the
most common way used in the literature for the uniform random generation of
contingency tables (see for
instance\cite{diaconnis-saloff,DBLP:journals/isci/Pardo94}). Mixing time
results have been provided for the particular case of $2\times n$ sized tables
in\cite{DBLP:journals/rsa/Hernek98} and the latter using a coupling argument
in\cite{DBLP:journals/tcs/DyerG00}. In this restricted context a
divide-and-conquer algorithm has recently been pointed out\cite{desalvo}.
In practice, there are MCMC dedicated packages for most common
programming languages:
{\tt mcmc}\footnote{\url{https://cran.r-project.org/web/packages/mcmc/index.html}}
for R, {\tt pymc}\footnote{\url{https://pypi.python.org/pypi/pymc/}}
for Python, \ldots

More generally, random generation is a natural way for performance evaluation used, for
instance in SAT-solver
competitions\footnote{http://www.satcompetition.org/}. In a distributed
computing context, it has been used for instance for the random generation
of DAG modelling graph task for parallel
environments\cite{10.1007/978-3-642-30111-7_93,DBLP:conf/simutools/CordeiroMPTVW10}.

%%%%%%%%%%%%%%%%%%%%%%%%%%%%%%%%%%%%%%%%%%%%%%%%%%%%%%%%%%%%%%%%%%%%%%%%%%
\section{Contingency vectors initialization}\label{section:vectors}

Considering $n$ tasks and $m$ machines, the first step in order to generate
instances is to fix the average cost of each task and the average speed of
each machine. Since $n$ and $m$ are fixed, instead of generating cost
averages, we generate the sum of the cost on each row and column, which is
related. The problem becomes, given $n,m$ and $N$ (total cost) to generate
randomly (and uniformly) two vectors $\omu\in \mathbb{N}^n$ and
$\onu\in \mathbb{N}^m$ satisfying:
\begin{equation}\label{eq:contengency1}
\sum_{i=1}^n\omu(i)=\sum_{j=1}^m \onu(j)=N,
\end{equation}
with the following convention on notations: for any vector
$\overline{v}=(v_1,\ldots,v_\ell)\in \mathbb{N}^\ell$, $v_i$ is denoted
$\overline{v}(i)$.

Moreover, the objective is also to limit the maximum value.
This is useful to avoid large variance: for this purpose we
restrict the generation to vectors whose parameters are in a controlled interval
$[\alpha,\beta]$. 
%Given, $n$, $m$ and $N$, the question of finding two decompositions $\omu$
%and $\onu$ satisfying Equation~(\ref{eq:contengency1}) 
This question is addressed in this
section using a classical recursive
approach\cite{DBLP:journals/tcs/FlajoletZC94}. More precisely, let
$\alpha\leq \beta$ be positive integers and $H_{N,n}^{\alpha,\beta}$ be the
subset of elements $\omu$ of $\mathbb{N}^n$ such that $N=\sum_{i=1}^n
\omu(i)$ and for all $1\leq i\leq n$, $\alpha\leq \omu(i)\leq \beta$ (i.e.\
the set of all possible vectors with values between $\alpha$ and
$\beta$). Let
$h_{N,n}^{\alpha,\beta}$ be the cardinal of $H_{N,n}^{\alpha,\beta}$. By
decomposition one has
\begin{equation}\label{eq:rec1}
h_{N,n}^{\alpha,\beta}=\sum_{k=\alpha}^\beta h_{N-k,n-1}^{\alpha,\beta}.
\end{equation}
Moreover,
\begin{equation}\label{eq:rec2}
\begin{array}{l}
h_{N,n}^{\alpha,\beta}=0 \text{ if } \alpha n < N\text{ or }   \beta n > N
\text{ and,}\\
h_{N,1}^{\alpha,\beta}=1  \text{ if } \alpha <  N < \beta.
\end{array}
\end{equation}

Algorithm~\ref{algo:nu} uniformly generates a random vector over
$H_{N,n}^{\alpha,\beta}$.

\begin{algorithm}
\DontPrintSemicolon
\KwIn{Integers $N$, $n$, $\alpha$, $\beta$}
\KwOut{$\omu\in \mathbb{N}^n$ such that $\alpha\leq\omu(i)\leq\beta$ and
  $\sum_i\omu(i)=N$ if it is possible\newline
$\bot$ otherwise.}

\Begin{
\If{$\alpha > \beta$ {\bf or} $n\alpha > N$  {\bf or} $n\beta < N$} {\Return
  $\bot$}
\For{$1\leq k\leq n$ {\bf and } $0\leq N^\prime \leq N$}{{\bf compute}
  $h_{N^\prime,k}^{\alpha,\beta}$ {\bf using (\ref{eq:rec1}) and (\ref{eq:rec2})}.}
\For{$i\in [1,\ldots,n]$}{
$s=0$\;
\eIf{$N-s\geq 0$}{
{\bf pick at random} $\omu(i)\in [\alpha,\beta]$ {\bf with }
$\Prob(\omu(i)=k)=\dfrac{h_{N-s-k,k}^{\alpha,\beta}}{h_{N-s,n-i}^{\alpha,\beta}}$\;
$s=s+k$}{$\omu_i=0$}
}
\Return{$\omu$}
}

\caption{Generate Sequences}\label{algo:nu}
\end{algorithm}

Note that integers involved in these computations may become rapidly very
large.
Working with floating point approximations to represent integers may
be more efficient.
Moreover, with the rounded errors the random generation stays very
close to the uniform distribution\cite{DBLP:journals/tcs/DeniseZ99}.

Figure~\ref{fig:distrib} depicts the distribution of the values when
varying the interval $[\alpha,\beta]$ for $n=10$ and $N=100$.
Without constraint ($\alpha=0$ and $\beta=100$), the distribution is
similar to an exponential one: small values are more likely to appear
in a vector than large ones.
When only the largest value is bounded ($\alpha=0$ and $\beta=15$),
then the shape of the distribution is inverted with smaller values
being less frequent.
Finally, bounding from both sides ($\alpha=5$ and $\beta=15$) leads to
a more uniform distribution.

\begin{figure}
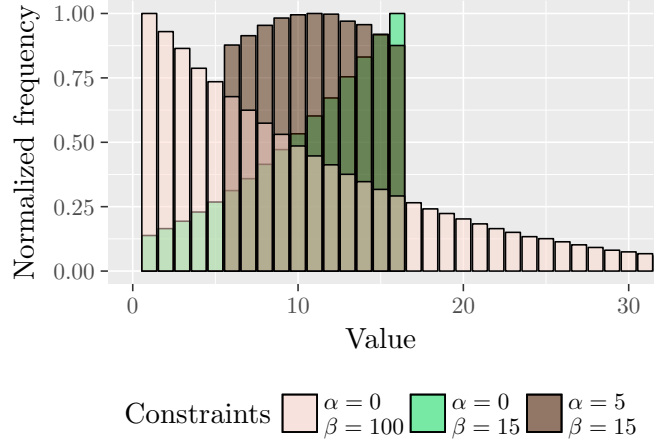

  \centering
  \small
  \includestandalone{2_hist_unif_vector}
  \caption{Frequency of each value in a vector of size $n=10$ with
    $N=100$ generated by Algorithm~\ref{algo:nu} for three
    combinations of constraints for the minimum $\alpha$ and maximum
    $\beta$.
    For each case, the frequencies were determined by generating
    $100\,000$ vectors and are normalized to the maximum frequency.
    The frequency for large values when $\alpha=0$ and $\beta=100$ are
    not shown.}
  \label{fig:distrib}
\end{figure}

Figure~\ref{fig:variance} shows the CV obtained for all possible
intervals $[\alpha,\beta]$.
The more constrained the values are, the lower the CV.
The CV goes from 0 when either $\alpha=10$ or $\beta=10$ (the vector
contains only the value 10) to 1 when $\alpha=0$ and $\beta=100$.

\begin{figure}
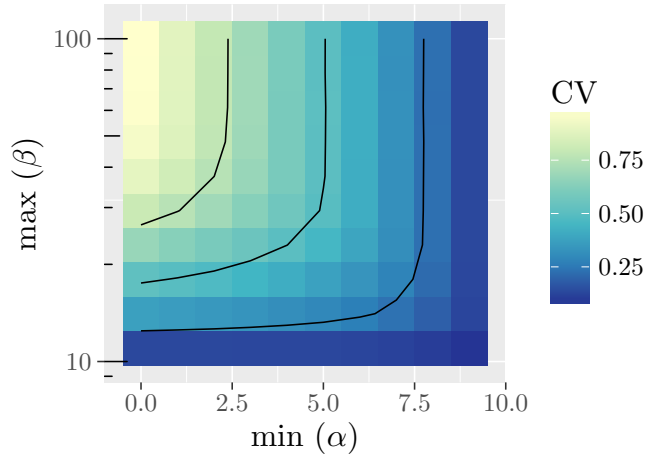

  \centering
  \includestandalone{3_heatmap_unif_vector}
  \caption{Mean CV in vectors of size $n=10$ with $N=100$ generated by
    Algorithm~\ref{algo:nu} for different constraints for the minimum
    $\alpha$ and maximum $\beta$.
    Each tile corresponds to $10\,000$ vectors.
    The contour lines correspond to the levels in the legend (2.5, 5,
    7.5 and 10).}
  \label{fig:variance}
\end{figure}

It is also possible to generate uniform vectors using Boltzmann
samplers\cite{DBLP:journals/cpc/DuchonFLS04}: this approach consists
in generating each $\onu(i)$ independently according to an exponential
law of parameter~$\gamma$.
Theoretical results of\cite{DBLP:journals/cpc/DuchonFLS04} show that
by choosing the right $\gamma$, the sum of the generated $\onu(i)$'s
is close to $N$ with a high probability.
In order to get precisely $N$, it suffices to use a rejection
approach.
This is consistent with the seemingly exponential distribution in
Figure~\ref{fig:distrib} in the unconstrained case.
Moreover, in this case, Figure~\ref{fig:variance} shows that the CV is
close to one, which is also the CV of an exponential distribution.

%%%%%%%%%%%%%%%%%%%%%%%%%%%%%%%%%%%%%%%%%%%%%%%%%%%%%%%%%%%%%%%%%%%%%%%%%%
\section{Symmetric Ergodic Markov Chains for the Random Generation}\label{section:MC}

We can now generate two random vectors $\omu$ and $\onu$ containing
the sum of each row and column with Algorithm~\ref{algo:nu}.
To obtain actual cost, we use Markov Chains to generate the
corresponding contingency table.
Random generation using finite discrete Markov Chains can easily be
explained using random walk on finite graphs.
Let $\Omega$ be the finite set of all possible cost matrices (also
called states) with given row and column sums: we want to sample
uniformly one of its elements.
However, $\Omega$ is too
large to be built explicitly. The approach consists in building a directed graph
whose set of vertices is $\Omega$ and whose set of edges represent all
the possible transitions between any pair of states.
Each edge of the graph is weighted by a
probability with a classical normalization: for each vertex, the sum of the
probabilities on outgoing edges is equal to~$1$. One can now consider random
walks on this graph. A classical Markov Chain result claims that for some
families of 
probabilistic graphs/Markov Chains,  walking long enough in the
graph, we have the same probability to be in each state, whatever the
starting vertex of the walk\cite[Theorem~4.9]{LevinPeresWilmer2006}.

This is the case for symmetric ergodic Markov Chains\cite[page 37]{LevinPeresWilmer2006}.
Symmetric means that if there is an edge $(x,y)$ with probability
$p$, then the graph has an  edge $(y,x)$ with the same
probability. A Markov Chain is ergodic if it is aperiodic (the gdc of the
lengths of loops of the graph is $1$) and if the graph is strongly
connected. When there is a loop of length $1$, the ergodicity issue reduces
to the strongly connected problem. In general, the graph is not explicitly
built and neighborhood relation is defined by a function, called a random
mapping, on each state. 
For a general reference on finite Markov Chains with many pointers,
see~\cite{LevinPeresWilmer2006}.

An illustration example is depicted on Fig~\ref{fig:exMC}. For
instance, starting arbitrarily from the central vertex, after one
step, we are in any other vertex with probability~$\frac{1}{6}$ (and
with probability $0$ in the central vertex since there is no self-loop
on it). After two steps, we are in the central vertex with
probability~$\frac{1}{6}$ and in any other with probability
$\frac{5}{36}$. In this simple example, one can show that after $n+1$
step, the probability to be in the central node is
$p_{n+1}=\frac{1}{7}(1-\left(\frac{-1}{6}\right)^n)$ and is
$\frac{1-p_{n+1}}{6}$ for all the other nodes. All probabilities tends
to $\frac{1}{7}$ when $n$ grows. 

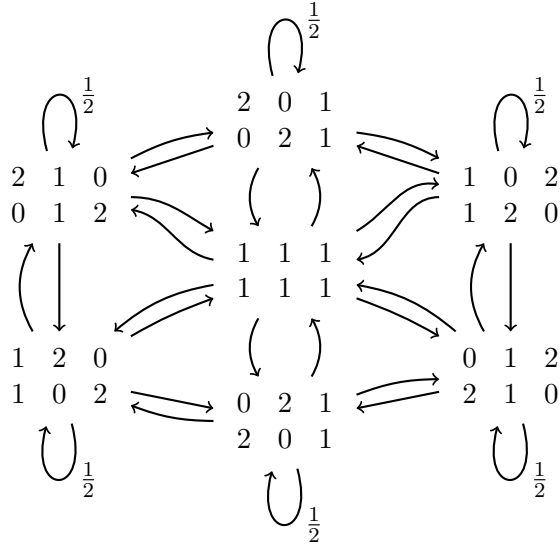
\begin{figure}
\centering
\begin{tikzpicture}

\node (C) at (0,0) {$\begin{array}{ccc}1 & 1 & 1\\1 & 1 & 1\end{array}$};
\node (A) at (0,2) {$\begin{array}{ccc}2 & 0 & 1\\0 & 2 & 1\end{array}$};
\node (D) at (3,1){$\begin{array}{ccc}1 & 0 & 2\\1 & 2 & 0\end{array}$};
\node (B) at (-3,1){$\begin{array}{ccc}2 & 1 & 0\\0 & 1 & 2\end{array}$};
\node (E) at (-3,-1.4){$\begin{array}{ccc}1 & 2 & 0\\ 1 & 0 & 2\end{array}$};
\node (F) at (3,-1.4){$\begin{array}{ccc}0 & 1 & 2\\ 2 & 1 & 0\end{array}$};
\node (G) at (0,-2){$\begin{array}{ccc}0 & 2 & 1\\ 2 & 0 & 1\end{array}$};

\path[->,thick] (C) edge[bend right] node[right] {} (A);
\path[->,thick] (A) edge[bend right] node[right] {} (C);

\path[->,thick] (C) edge[bend right] node[right] {} (G);
\path[->,thick] (G) edge[bend right] node[right] {} (C);

\path[<-,thick] (C) edge[out=150,in=0] node[above] {} (B);
\path[<-,thick] (B) edge[out=-10,in=170] node[left] {} (C);

\path[->,thick] (D) edge[in=10,out=180] node[right] {} (C);
\path[->,thick] (C) edge[in=170,out=30] node[above] {} (D);

\path[->,thick] (E) edge[out=30,in=200] node[below] {} (C);
\path[->,thick] (C) edge[out=190, in=40] node[above] {} (E);

\path[->,thick] (F) edge[in=-10,out=140] node[above] {} (C);
\path[->,thick] (C) edge[out=-20, in = 150] node[below] {} (F);

\path[->,thick] (B) edge[] node[right] {} (E);
\path[->,thick] (E) edge[bend left] node[right] {} (B);

\path[->,thick] (D) edge[] node[right] {} (F);
\path[->,thick] (F) edge[bend left] node[right] {} (D);

\path[->,thick] (E) edge[] node[above] {} (G);%
\path[->,thick] (G) edge[bend left=10] node[below] {} (E);%

\path[->,thick] (F) edge[] node[above] {} (G);%
\path[->,thick] (G) edge[bend left=10] node[below] {} (F);%

\path[->,thick] (A) edge[pos=0.2] node[below] {} (B);
\path[->,thick] (B) edge[pos=0.2,bend left=10] node[above] {}
(A);

\path[->,thick] (D) edge[pos=0.8] node[below] {} (A);
\path[->,thick] (A) edge[pos=0.2,bend left=10] node[above] {} (D);

\path[->,thick] (A) edge[loop above] node[right] {$\ \frac{1}{2}$} ();
\path[->,thick] (B) edge[loop above] node[right] {$\ \frac{1}{2}$} ();
\path[->,thick] (D) edge[loop above] node[right] {$\ \frac{1}{2}$} ();
\path[->,thick] (E) edge[loop below] node[right] {$\ \frac{1}{2}$} ();
\path[->,thick] (F) edge[loop below] node[right] {$\ \frac{1}{2}$} ();
\path[->,thick] (G) edge[loop below] node[right] {$\ \frac{1}{2}$} ();
\end{tikzpicture}
\caption{Example of the underlying graph of a Markov Chain when the
  sum of each row is three and the sum of column is two.
  Unless otherwise stated, each transition probability is
  $\frac{1}{6}$.}\label{fig:exMC}
\end{figure}

This section is dedicated to building symmetric and ergodic Markov Chains for
our problem. In Section~\ref{sec:omega} we define the sets $\Omega$ that are
interesting for cost matrices.  In Section~\ref{sec:MC}, Markov Chains are
proposed using a dedicated random mapping and are proved to be symmetric and
ergodic. Finally, in Section~\ref{sec:rapid} we use classical techniques to
transform the Markov Chains into other symmetric ergodic MC mixing faster
(i.e.\ the number of steps required to be close to the uniform distribution
is smaller). 

Recall that $N,n,m$ are positive integers and that
$\omu\in \mathbb{N}^n$ and $\onu\in \mathbb{N}^m$ satisfy
Equation~(\ref{eq:contengency1}).

%%%%%%%%%%%%%%%%
\subsection{Contingency Tables}\label{sec:omega}

In this section, we define the state space of the Markov Chains. We
consider contingency tables with fixed sums on rows and columns. We also
introduce min/max constraints in order to control the variance of the value. 
We denote by $\Om$ the set of positive $n\times
m$ matrices $M$
over $\mathbb{N}$ such that for every $i\in\{1,\ldots,n\}$ and every
$j\in\{1,\ldots,m\}$,
\begin{equation}\label{eq:contengency2}
\sum_{k=1}^m M(i,k)=\omu(i)\quad\text{and}\quad\sum_{k=1}^n M(k,j)=\onu(j)
\end{equation}For example, the matrix
$$M_\mathrm{exa}=\left(\begin{array}{ccc}3 & 1 \\ 2 & 0 \\ 5 & 10 \end{array}\right)$$
is in $\Omega_{2,3}(\omu_\mathrm{exa},\onu_\mathrm{exa})$, where $\omu_\mathrm{exa}=(4,2,15)$ and
$\onu_\mathrm{exa}=(10,11)$.
\medskip

The first restriction consists in having a global minimal value $\alpha$ and
a maximal global value $\beta$ on the considered matrices.  Let
$\alpha,\beta$ be positive integers. We denote by
$\Om[\alpha,\beta]$ the subset of
$\Om$ of matrices $M$ such that for all $i,j$,
$\alpha\leq M(i,j)\leq \beta$. For example, $M_\mathrm{exa}\in
\Omega_{2,3}(\omu_\mathrm{exa},\onu_\mathrm{exa})[0,12]$. If $\beta < \alpha$, then
$\Om[\alpha,\beta]=\emptyset$. Moreover, according to
Equation~(\ref{eq:contengency2}), one has
\begin{equation}\label{eq:maxglobal}
\begin{array}{ll}
\Om&=\Om[0,N]\\&=\Om[0,\min(\max_{1\le
    k\le m}\omu(k),\\&\hspace{3em}\max_{1\le k\le n}\onu(k))].
\end{array}
\end{equation}

Now we consider min/max constraints on each row and each line.
Let $\oalpha_c,\obeta_c\in \mathbb{N}^m$ and $\oalpha_r,\obeta_r\in
\mathbb{N}^n$. We denote by
$\Om[\oalpha_c,\obeta_c,\oalpha_r,\obeta_r]$ the
subset of $\Om$ of matrices $M$ satisfying: for all
$i,j$, $\oalpha_c(j)\leq M(i,j)\leq \obeta_c(j)$ and $\oalpha_r(i)\leq
M(i,j)\leq \obeta_r(i)$.
For instance, $$M_{\rm
  exa}\in \Omega_{2,3}(\omu_\mathrm{exa},\onu_\mathrm{exa})[(1,0,5),(3,2,10),(2,0),(5,10)].$$
Using Equation~(\ref{eq:contengency2}), one has for
every $\alpha,\beta\in\mathbb{N}$,
\begin{equation}\label{eq:maxrowcol}
\begin{array}{l}
\Om[\alpha,\beta]=\Om[(\alpha,\ldots,\alpha),\\\hspace{3em}(\beta,\ldots,\beta),(\alpha,\ldots,\alpha),(\beta,\ldots,\beta)].
\end{array}
\end{equation}

To finish, the more general constrained case, where min/max are defined for
each element of the matrices. 
Let $\Am$ and $\Bm$ be two $n\times m$ matrices of positive integers.
We denote by
$\Om[Am,Bm]$ the subset of $\Om$
of matrices $M$ such that for all $i,j$, $\Am(i,j)\leq M(i,j)\leq \Bm(i,j)$.
For instance, one has $M_\mathrm{exa}\in \Om[A_{\rm
    exa},B_\mathrm{exa}]$, with
$$A_{\rm
  exa}=\left(\begin{tabular}{ccc}3&2&4\\0&0&5\end{tabular}\right)\quad
  \text{and}\quad B_{\rm
  exa}=\left(\begin{tabular}{ccc}5&4&6\\1&3&12\end{tabular}\right).
$$
 For
every $\oalpha_c,\obeta_c\in\mathbb{N}^m,\oalpha_r,\obeta_r\in\mathbb{N}^n$,
one has
\begin{equation}\label{eq:maxelement}
\Om[\oalpha_c,\obeta_c,\oalpha_r,\obeta_r]=\Om[A,B],
\end{equation}
where $A(i,j)=\max\{\oalpha_c(j),\oalpha_r(i)\}$ and
$B(i,j)=\min\{\obeta_c(j),\obeta_r(i)\}$.

%%%%%%%%%%%%%%%%
\subsection{Markov Chains}\label{sec:MC}

As explained before, the random generation process is based on symmetric
ergodic Markov Chains. This section is dedicated to define such chains on
state spaces of the form $\Om$, $\Om[\alpha,\beta]$,
$\Om[\oalpha_c,\obeta_c,\oalpha_r,\obeta_r]$ and $\OMEGA$. According
to Equations~(\ref{eq:maxglobal}), (\ref{eq:maxrowcol})
and~(\ref{eq:maxelement}), it suffices to work on $\OMEGA$. To
simplify the notation, let us denote by $\Omega$ the set $\OMEGA$.

For  any $1\leq i_0,i_1\leq n$, any $1\leq j_0,j_1,\leq m$, such that
$i_0\neq i_1$ and $j_0\neq j_1$,
we denote by $\Delta_{i_0,i_1,j_0,j_1}$ the $n\times m$ matrix defined by
$\Delta(i_0,j_0)=\Delta(i_1,j_1)=1$, $\Delta(i_0,j_1)=\Delta(i_1,j_0)=-1$,
and $\Delta(i,j)=0$ otherwise. For instance, for $n=3$ and $m=4$ one has
$$\Delta_{1,2,1,3}=\left(\begin{tabular}{cccc} 1 & 0 & -1 & 0\\
-1 & 0 & 1 & 0 \\ 0 & 0 & 0 & 0
\end{tabular}\right).$$

Tuple $(i_0,j_0,i_1,j_1)$ is used as follow to shuffle a cost matrix
and to transit from one state to another in the markov chain:
$\Delta_{i_0,i_1,j_0,j_1}$ is added to the current matrix, which
preserves the row and column sums.
Formally, let $K=\{(i_0,j_0,i_1,j_1)\mid i_0\neq i_1, j_0\neq j_1,
1\leq i_0,i_1\leq
n, 1\leq j_0,j_1\leq m\}$ be the set of all possible tuples.
Let $f$ be the mapping function from $\Omega\times
K$ to $\Omega$ defined by
$f(M,(i_0,j_0,i_1,j_1))=M+\Delta_{(i_0,j_0,i_1,j_1)}$ if
$M+\Delta_{(i_0,j_0,i_1,j_1)}\in \Omega$ and $M$ otherwise.
The mapping is called at each iteration, changing the instance
until it is sufficiently shuffled.

We consider the Markov chain $\mathcal{M}$ defined on $\Omega$ by the
random mapping $f(\cdot,U_K)$, where $U_K$ is a uniform random variable on
$K$.

The following result gives the properties of the markov chain and
is an extension of a similar
result\cite{diaconnis-saloff} on $\Om$. The difficulty is to prove that the
underlying graph is strongly connected since the constraints are hindering
the moves.

\begin{theorem}\label{thm:ergodic}
The Markov Chain  $\mathcal{M}$ is symmetric and ergodic.
\end{theorem}

The proof of Theorem~\ref{thm:ergodic} is based on Lemma~\ref{lm:stair}
and~\ref{lm:path}.

\begin{definition}\label{def:stair}
Let $A$ et $B$ be two elements of $\Omega$. A finite sequence
$u_1=(i_1,j_1),\ldots,u_r=(i_r,j_r)$ of pairs of indices in
$\{1,\ldots,n\}\times\{1,\ldots,m\}$ is called a {\it stair} sequence for
$A$ and $B$ if it satisfies the following properties:
\begin{enumerate}
\item $r \geq 4$,
\item If $k\neq \ell$, then $u_k\neq u_\ell$,

\item If $1\le k<r$ is even, then $j_k=j_{k+1}$ and $A(i_k,j_k) < B(i_k,j_k)$
\item If $1\le k<r$ is odd, then $i_k=i_{k+1}$ and $A(i_k,j_k) > B(i_k,j_k)$,

\item $r$ is even and  $j_r=j_1$,
%\item $r$ is odd and $i_r=i_1$.
\end{enumerate}
\end{definition}

Consider, for instance, the matrices 
$$A_1=\left(
\begin{array}{ccccc}
3 & 0 & 0 & 0 & 7\\
7 & 4 & 0 & 0 & 0\\
0 & 7 & 5 & 0 & 0\\
0 & 0 & 7 & 6 & 0\\
0 & 0 & 0 & 7 & 5\\
\end{array}
\right)
B_1=\left(
\begin{array}{ccccc}
2 & 1 & 0 & 0 & 7\\
7 & 3 & 1 & 0 & 0\\
0 & 7 & 4 & 1 & 0\\
0 & 0 & 7 & 5 & 1\\
1 & 0 & 0 & 7 & 4\\
\end{array}
\right).
$$ 

The sequence $(1,1), (1,2),(2,2),(2,3),(3,3),(3,4),(4,4),\\ (4,5),(5,5),(5,1)$ is a
stair sequence for $A_1$ and $B_1$. 

\begin{lemma}\label{lm:stair}
Let $A$ et $B$ be two distinct elements of $\Omega$. There
exists a stair sequence for $A$ and $B$.
\end{lemma}

\begin{proof}
The proof is by construction. Since $A$ and $B$ are distinct, using the
constraints on the sums of rows and columns, there exists a pair of indices
$u_1=(i_1,j_1)$ such that $A(i_1,j_1) > B(i_1,j_1)$. Now using the sum
constraint on row  $i_1$, there exists $j_2$ such that $B(i_1,j_2) <
A(i_1,j_2)$. Set $u_2=(i_1,j_2)$. Similarly, using the sum constraint on
column $j_2$, there exists $i_3\neq i_1$ such that $A(i_3,j_2) > B(i_3,j_2)$.
Set $u_3=(i_3,j_2)$. Similarly, by the constraint on row $i_3$, there
exists $j_4$ such that $A(i_3,j_4) < B(i_3,j_4)$. At this step,
$u_1,u_2,u_3,u_4$ are pairwise distinct.

If $j_4=j_1$, then $u_1,u_2,u_3,u_4$ is a stair sequence for $A$ and $B$.
Otherwise, by the $j_4$-column constraint, there exists $i_5$ such that
$B(i_5,j_4) > A(i_5,j_4)$. Now, one can continue the construction until the
first step $r$ we get either  $i_r=i_s$ or $j_r=j_s$ with $s <r$ (this step
exists since the set of possible indexes is finite). Note that we consider
the smallest $s$ for which this is case.
\begin{itemize}
\item If $i_r=i_s$, $s <r$, the sequence $u_1,u_2,\ldots,u_r$ satisfies the
  conditions 2., 3. and 4. of Definition~\ref{def:stair}. Moreover both $r$
  and $s$ are odd. The sequence $u_s,\ldots,u_r$ satisfies the Conditions 2.
  to 5. of Definition~\ref{def:stair}. Since $r >s$ and by construction,
  $r-s > 4$. If follows that the sequence $u_r,u_{r-1},\ldots,u_{s+1}$ is a
  stair sequence for $A$ and $B$.
\item If $j_r=j_s$, then both $r$ and $s$ are even. The sequence
  $u_{s+1},\ldots,u_r$ satifies the Conditions 1. to 5. of
  Definition~\ref{def:stair} and is therefore a stair sequence for $A$ and
  $B$.

\end{itemize}
\end{proof}

Given two $n \times m$ matrices $A$ and $B$, the distance from $A$ to
$B$, denoted $d(A,B)$, is defined by:
\[
  d(A,B)=\sum_{i=1}^n\sum_{j=1}^m |A(i,j)-B(i,j)|.
\]

\begin{lemma}\label{lm:path}
Let $A$ et $B$ be two distinct elements of $\Omega$. There
exists $C\in \Omega$ such that $d(C,B) < d(A,B)$ and
tuples $t_1,\ldots,t_k$ such that
$ C=f(\ldots f(f(A,t_1),t_2)\ldots, t_k)$ and for every
  $\ell\leq k$, $f(\ldots f(f(A,t_1),t_2)\ldots, t_\ell)\in
  \Omega$.
\end{lemma}

\begin{proof}
By Lemma~\ref{lm:stair}, there exists a stair sequence $u_1,\ldots,u_r$ for $A$ and
$B$. Without loss of generality (using a permutation of rows and columns) we
may assume that $u_{2k+1}=(k,k)$ and $u_{2k}=(k,k+1)$, for $k < \frac{r}{2}$
and $u_r=(\frac{r}{2},1)$.

To illustrate the proof, we introduce some $\dfrac{r}{2}\times \dfrac{r}{2}$ matrix $M$
over $\{+,-,\min,\max\}$, called {\it difference matrices}, such that: if
$M(i,j)=+$, then $A(i,j) > B(i,j)$; if $M(i,j)=-$, then $A(i,j) <
B(i,j)$; if $M(i,j)=\min$, then $A(i,j)=\Am(i,j)$; and if $M(i,j)=\max$,
then $A(i,j)=\Bm(i,j)$.

Considering for instance the matrices $A_1$ and $B_1$ defined before, with a
global minimum equal to $0$ and global maximum equal to $7$, a
difference matrix is

$$\left(
\begin{array}{ccccc}
+ & - & \min & \min & \max\\
\max & + & - & \min & \min\\
\min & \max & + & - & \min\\
\min & \min & \max & + & -\\
- & \min & \min & \max & +\\
\end{array}
\right).
$$ 
Note that it may exist several difference matrices since, for instance, some $+$ might be
replaced by  a $\max$ 

The proof investigates several cases:
\begin{itemize}
\item[Case 0:] If $r =4$, then $k=1$ and $t_1=(2,1,1,2)$ works. Indeed,
  since $\Bm(i,j)\geq A(1,1)>B(1,1)\geq \Am(i,j)$, one has $\Am(1,1)\leq
  A(1,1)-1 \leq \Bm(1,1)$. Similarly, $\Am(2,1)\leq A(2,1)+1 \leq
  \Bm(2,1)$, $\Am(1,2)\leq A(1,2)+1 \leq \Bm(1,2)$ and $\Am(2,2)\leq
  A(2,2)-1 \leq \Bm(2,2)$. It follows that $C=f(A,(2,1,1,2))\in\Omega$ and
  $d(C,B)=d(A,B)-4 < d(A,B)$. In this case, the following matrix is a
  difference matrix:
$$
\left(
\begin{array}{cc}
+ & - \\
- & +
\end{array}
\right).
$$

\item[Case 1:] If $r > 4$ and if there exists $3 \leq \ell \leq \frac{r}{2}$
  such that $A(\ell-2,\ell)\neq \Am(\ell-2,\ell)$, then, as for Case~0, $k=1$
  works with $t_1=(\ell-2,\ell-1,\ell-1,\ell)$: $f(A,t_1)\in\Omega$.
  Moreover $d(f(A,t_1),B)=d(A,B)-4$ if $A(\ell-2,\ell)>B(\ell-2,\ell)$;
  $d(f(A,t_1),B)=d(A,B)-2$ otherwise. In this case, the following matrix is a
  difference matrix:
$$
\left(
\begin{array}{ccccccc}
+ & - &&&&&\\
& + & - &A(\ell-2,\ell)&&&\\
 && + & - &&&\\
 &&& + & \ddots &&\\
 &&&& \ddots & - &\\
 &&&&& + & -\\
 -&&&&&& +  \\
\end{array}
\right).
$$
\item[Case 2:] If $r > 4$ and Case~1 does not hold and if there exists $1\leq
  \ell \leq \frac{r}{2}-1$ such that $A(\ell+1,\ell)< \Bm(\ell+1,\ell)$,
  then, similarly, $k=1$ and $t_1=(\ell,\ell+1,\ell+1,\ell)$ works.
  One has $d(f(A,t_1),B)=d(A,B)-4$ if $A(\ell+1,\ell)<
  B(\ell+1,\ell)$, and $d(f(A,t_1),B)=d(A,B)-2$ otherwise.
  In this
  case, the following matrix is a difference matrix:
$$
\left(
\begin{array}{ccccccc}
+ & - &\min&&&&\\
& + & - &\min&&&\\
 &A(\ell+1,\ell)& + & - &\min&&\\
 &&& + & \ddots &\ddots&\\
 &&&& \ddots & - &\min\\
 &&&&& + & -\\
 -&&&&&& +  \\
\end{array}
\right).
$$

\item[Case 3:] If Cases~0 to~2 do not hold and if $A(1,\frac{r}{2})\neq
  \Bm(1,\frac{r}{2})$, then, similarly, $k=1$ and
  $t_1=(1,\frac{r}{2},\frac{r}{2},1)$ works. One has $d(f(A,t_1),B)=d(A,B)-4$
  if $A(1,\frac{r}{2})< B(1,\frac{r}{2})$, and $d(f(A,t_1),B)=d(A,B)-2$
  otherwise. In this case, the following matrix is a difference matrix:
$$
\left(
\begin{array}{ccccccc}
+ & - &\min&&&&A(1,\frac{r}{2})\\
\max& + & - &\min&&&\\
 &\max& + & - &\min&&\\
 &&\max& + & \ddots &\ddots&\\
 &&& \ddots & \ddots & - &\min\\
 &&&&\max& + & -\\
 -&&&&&\max& +  \\
\end{array}
\right).
$$

\item[Case 4:] If Cases~0 to~3 do not hold. Since
  $A(\frac{r}{2},\frac{r}{2})=\Bm(1,\frac{r}{2})$ and
  $A(1,\frac{r}{2}-2)=\Am(\frac{r}{2},\frac{r}{2}-2)$, $i_0=\max \{i\mid 1\leq i <
  \frac{r}{2}-2 \text{ and } A(i,\frac{r}{2})\neq \Am(i,\frac{r}{2})\}$
  exists. In this case $t_1=(i_0,i_0+1,i_0+1,\frac{r}{2}),
  t_2=(i_0+1,i_0+2,i_0+2,\frac{r}{2}),\ldots,
  t_{\frac{r}{2}-2-i_0}=(\frac{r}{2}-2,\frac{r}{2}-1,\frac{r}{2}-1,\frac{r}{2})$
  works.
  With $C=f(\ldots f(f(A,t_1),t_2)\ldots,t_{\frac{r}{2}-2-i_0})$.
  One has $d(C,B)=d(A,B)-2\times(\frac{r}{2}-i_0-1)$ if
  $A(i_0,\frac{r}{2})>B(i_0,\frac{r}{2})$, and
  $d(C,B)=d(A,B)-2\times(\frac{r}{2}-i_0-2)$ otherwise.
  Moreover, for every $\ell\leq \frac{r}{2}-2-i_0$, $f(\ldots
  f(f(A,t_1),t_2)\ldots, t_\ell)\in \Omega$.
  In this case, the following matrix is a difference matrix:
$$
\left(
\begin{array}{ccccccc}
+ & - &\min&&&&\max\\
\max& + & - &\min&&&\\
 &\max& + & - &\min&&A(i_0,\frac{r}{2})\\
 &&\max& + & \ddots &\ddots&\vdots\\
 &&& \ddots & \ddots & - &\min\\
 &&&&\max& + & -\\
 -&&&&&\max& +  \\
\end{array}
\right).
$$
\end{itemize}

\end{proof}

One can now prove Theorem~\ref{thm:ergodic}.

\begin{proof}
If $A=f(B,(i_0,j_0,i_1,j_1))$, then $B=f(A,(i_1,j_1,i_0,j_0))$, proving that
the Markov Chain is symmetric.

Let $A_0\in \Omega$. We define the sequence $(A_k)_{k\geq 0}$ by
$A_{k+1}=f(A_k,(1,1,2,2))$. The sequence $A_k(1,2)$ is decreasing and
positive. Therefore, one can define the smallest index $k_0$ such that
$A_{k_0}(1,2)=A_{k_0+1}(1,2)$. By construction, one also has
$A_{k_0}=A_{k_0+1}$. It follows that the Markov Chain is aperiodic.

Since $d$ is a distance, irreducibility is a direct consequence of Lemma~\ref{lm:path}.
\end{proof}

Consider the two matrices $A_1$ and $B_1$ defined previously with
$\Bm$ containing only the value 7.
Case~4 of the
proof can be applied. One has $t_1=(1,2,2,5)$ and
$$
f(A_1,t_1)=A_2=\left(
\begin{array}{ccccc}
3 & 1 & 0 & 0 & 6\\
7 & 3 & 0 & 0 & 1\\
0 & 7 & 5 & 0 & 0\\
0 & 0 & 7 & 6 & 0\\
0 & 0 & 0 & 7 & 5\\
\end{array}
\right).
$$
Next, $t_2=(2,3,3,5)$ and
$$
f(A_2,t_2)=A_3=\left(
\begin{array}{ccccc}
3 & 1 & 0 & 0 & 6\\
7 & 3 & 1 & 0 & 0\\
0 & 7 & 4 & 0 & 1\\
0 & 0 & 7 & 6 & 0\\
0 & 0 & 0 & 7 & 5\\
\end{array}
\right).
$$
We have $t_3=(3,4,4,5)$ and
$$
f(A_3,t_3)=A_4=\left(
\begin{array}{ccccc}
3 & 1 & 0 & 0 & 6\\
7 & 3 & 1 & 0 & 0\\
0 & 7 & 4 & 1 & 0\\
0 & 0 & 7 & 5 & 1\\
0 & 0 & 0 & 7 & 5\\
\end{array}
\right).
$$
Finally, $f(A_4,(5,1,1,5))=B_1$ (Case 0): there is a path from $A_1$
to $B_1$ and,
since the chain is symmetric, from $B_1$ to $A_1$.

%%%%%%%
\subsection{Rapidly Mixing Chains}\label{sec:rapid}

The chain $\mathcal{M}$ can be classically modified in order to mix faster: once an
element of $\mathcal{K}$ is picked up, rather than changing each element by
$+1$ or $-1$, each one is modified by $+a$ or $-a$, where $a$ is picked
uniformly in order to respect the constraints of the matrix. This approach,
used for instance in\cite{DBLP:journals/tcs/DyerG00}, allows moving faster,
particularly for large $N$'s. 

Moving in $\Om$, from matrix $M$, while $(i_0,j_0,i_1,j_1)$ has been picked
in $\mathcal{K}$, $a$ is uniformly chosen such that $a\leq
\min\{M(i_0,j_1),M(i_1,j_0)\}$ in order to keep positive elements in the
matrix. It can be generalized for constrained Markov Chains. For instance,
in $\Om[\alpha,\beta]$, one has $a\geq 1$ and 
\begin{align*}
a\leq \min\{\alpha - M(i_0,j_0),\alpha-M(i_1,j_1),\\M(i_0,j_1)-\beta,M(i_1,j_0)-\beta\}.
\end{align*}

This approach is used in the experiments described in
Sections~\ref{MC:convergence} and~\ref{sec:const-effect-cost}.

%%%%%%%%%%%%%%%%%%%%%%%%%%%%%%%%%%%%%%%%%%%%%%%%%%%%%%%%%%%%%%%%%%%%%%%%%%
\section{Convergence of the Markov Chains}\label{MC:convergence}

We can now generate a matrix that is uniformly distributed when the
Markov Chain is run long enough to reach a stationary distribution.
The mixing time $t_\mathrm{mix}(\varepsilon)$ of an ergodic Markov Chain is the
number of steps required in order to be $\varepsilon$-close to the
stationary distribution (for the total variation distance,
see\cite[Chapter~4]{LevinPeresWilmer2006}). Computing theoretical bounds on
mixing time is a hard theoretical problem. For two rowed contingency tables,
it is proved in\cite{DBLP:journals/tcs/DyerG00} that $t_{\rm
  mix}(\varepsilon)$ is in $O(n^2\log (\frac{N}{\varepsilon}))$ and
conjectured that it is in $\Theta(n^2\log (\frac{n}{\varepsilon}))$. The results
are extended and improved in\cite{Dyer06} for a fixed number of rows. As far
as we know, there are no known results for the general case. A frequently used
approach to tackle the convergence problem (when to stop mixing the chain)
consists in using statistical test. Starting from  a different point of the
state space (ideally well spread in the graph), we perform several random
walks and we monitor numerical parameter in order to observe the
convergence. For our work, used parameters are defined in
Section~\ref{sec:measures}. Section~\ref{sec:initial} is dedicated to
finding different starting points. Convergence experimental results are
given in Section~\ref{sec:xpconv}.

%%%%%%%%%%%%%%%%%%%%%%%%%%%%%%%%%%%%%%%%%%%%%%%%%%%%%%%%%%%%%%%%%%%%%%%%%%
\subsection{Measures}\label{sec:measures}

We apply a set of measures on the matrix at each step of the Markov
process to assess its convergence.
At first, these measures heavily depend on the initial matrix.
However, they eventually converge to a stationary distribution as the
number of steps increases.
In the following, we assume that once they converge, the Markov
Chain is close to the stationary distribution.

These measures consist in:
\begin{itemize}
\item the cost Coefficient-of-Variation (ratio of standard deviation
  to mean):\\
  $\sqrt{\frac1{nm}\sum_{i=1}^n\sum_{j=1}^m\left( M(i,j)/\frac{N}{nm}
      - 1 \right)^2}$
\item the mean of row Coefficients-of-Variation:
  $\sum_{i=1}^n\frac{\sqrt{\frac1{m}\sum_{j=1}^m(M(i,j)-\frac{\omu(i)}{m})^2}}{n\omu(i)}$
\item the mean of column Coefficients-of-Variation:
  $\sum_{j=1}^m\frac{\sqrt{\frac1{n}\sum_{i=1}^n(M(i,j)-\frac{\onu(j)}{n})^2}}{m\onu(j)}$
\item Pearson's $\chi^2$ statistic:
  $\sum_{i=1}^n\frac{(M(i,j)-\omu(i)\onu(j)/N)^2}{\omu(i)\onu(j)/N}$
\item the mean of row correlations:
  $\frac{1}{n(n-1)/2}\sum_{i=1}^{n-1}\sum_{i'=i+1}^n\rho(M(i,.),
  M(i',.))$
\item the mean of column correlations:
  $\frac{1}{m(m-1)/2}\sum_{j=1}^{m-1}\sum_{j'=j+1}^m\rho(M(.,j),
  M(.,j'))$
\end{itemize}
where $\rho(M(i,.), M(i',.))$ (resp.\ $\rho(M(.,j), M(.j'))$) denotes
the Pearson coefficient of correlation between rows $i$ and $i'$
(resp.\ columns $j$ and $j'$).
When a row or column contains identical values, the related
correlations are undefined.
When a row (resp.\ column) sum is zero, the mean of row (resp.\
column) CV and the $\chi^2$ are undefined.

The first measure is an indicator of the overall variance of the
costs.
The second two measures indicate whether this variance is distributed
on the rows (task heterogeneity) or the columns (machine
heterogeneity).
The $\chi^2$ is used to assess the proportionality of the costs and
the correlations show whether rows or columns are proportional.

%%%%%%%%%%%%%%%%%%
\subsection{Initial Matrix}\label{sec:initial}

The Markov Chain described in Section~\ref{section:MC} requires an
initial matrix.
Before reaching the stationary distribution, the Markov Chain
iterates on matrices with similar characteristics to the initial
one.
However, after enough steps, the Markov Chain eventually
converges.
We are interested in generating several initial matrices with
different characteristics to assess this number of steps.
Formally, given $\omu$, $\onu$, $\Am$ and
$\Bm$, how to find an element of $\OMEGA$ to start the Markov Chain?

We identify three different kinds of matrices for which we propose simple
generation methods:
\begin{itemize}
\item a \emph{homogeneous} matrix with smallest cost CV
  (Algorithm~\ref{algo:homo})
\item a \emph{heterogeneous} matrix with largest cost CV
  (Algorithm~\ref{algo:hetero})
\item a \emph{proportional} matrix with smallest Pearson's
  $\chi^2$ statistic (Algorithm~\ref{algo:prop})
\end{itemize}

Ideally, initial matrices could be generated with an exact method (e.g.\
with an integer programming solver).
However, the optimality is not critical to assess the time to converge
and Algorithms~\ref{algo:homo} to~\ref{algo:prop} have low costs but
are not guaranteed.

Moreover, the convergence may be the longest when the search space is
the largest, which occurs when the space is the least constrained.
Thus, Algorithms~\ref{algo:homo} to~\ref{algo:prop} are used to study
convergence without constraints $\Am$ and $\Bm$.
Only Algorithm~\ref{algo:prop} supports such constraints and is used
to study their effects in Section~\ref{sec:const-effect-cost}.

\begin{algorithm}
\DontPrintSemicolon
\KwIn{Integer vectors $\omu$, $\onu$}
\KwOut{$M\in\Om$}
\Begin{
  $M\gets\{0\}_{1\le i\le n,1\le j\le m}$\;
  \While{$\sum_{i=1}^n\sum_{j=1}^nM(i,j)\ne N$}{
    \eIf{$\max(\omu(i))/m\ge\max(\onu(j))/n$}{
      $i\gets\arg\max_{i}\omu(i)$\;
      sort $j_1,\ldots,j_m$ such that $\onu(j_k)\le\onu(j_{k+1})$\;
      \For{$j_k\in\{j_1\ldots,j_m\}$}{
        $d\gets\min(\onu(j_k), \frac{\omu(i)}{m-k+1})$\;
        $M(i,j_k)\gets M(i,j_k)+d$\;
        $\omu(i)\gets\omu(i)-d$\;
        $\onu(j_k)\gets\onu(j_k)-d$\;
      }
    }{
      perform the same operation on the transpose matrix (swapping
      $\omu$ and $\onu$)\;
    }
  }
  \Return{$M$}
}
\caption{Homogeneous Matrices\label{algo:homo}}
\end{algorithm}

\begin{algorithm}
\DontPrintSemicolon
\KwIn{Integer vectors $\omu$, $\onu$}
\KwOut{$M\in\Om$}
\Begin{
  $M\gets\{0\}_{1\le i\le n,1\le j\le m}$\;
  \While{$\sum_{i=1}^n\sum_{j=1}^nM(i,j)\ne N$}{
    $D\gets\min(\omu^T\cdot\mathbb{1}_m,\mathbb{1}_n^T\cdot\onu)$\;
    $i_{\max},j_{\max}\gets\arg\max_{i,j}D(i,j)$\;
    $d\gets D(i_{\max},j_{\max})$\;
    $M(i_{\max},j_{\max})\gets d$\;
    $\omu(i_{\max})\gets\omu(i_{\max})-d$\;
    $\onu(j_{\max})\gets\onu(j_{\max})-d$\;
  }
  \Return{$M$}
}
\caption{Heterogeneous Matrices\label{algo:hetero}}
\end{algorithm}

\begin{algorithm}
\DontPrintSemicolon
\KwIn{Integer vectors $\omu$, $\onu$, integer matrices $\Am$, $\Bm$}
\KwOut{$M\in\OMEGA$}
\Begin{
  $M\gets \max(\Am,\min(\lfloor\omu^T\times\onu/N+1/2\rfloor,\Bm))$\;
  $\omu(i)\gets\omu(i)-\sum_{j=1}^m M(i,j)$ for $1\le i\le n$\;
  $\onu(j)\gets\onu(j)-\sum_{i=1}^n M(i,j)$ for $1\le j\le m$\;
  \While{$\sum_{j=1}^mM(i,j)\ne\omu(i)$ or $\sum_{i=1}^nM(i,j)\ne\onu(j)$}{
    choose random $i$ and $j$\;
    $d\gets0$\;
    \lIf{$M(i,j)<\Bm(i,j),(\omu(i)>0\ \mathrm{or}\ \onu(j)>0)$}{$d\gets1$}
    \lIf{$M(i,j)>\Am(i,j),(\omu(i)<0\ \mathrm{or}\ \onu(j)<0)$}{$d\gets-1$}
    $M(i,j)\gets M(i,j)+d$\;
    $\omu(i)\gets\omu(i)-d$\;
    $\onu(j)\gets\onu(j)-d$\;
  }
  \Return{$M$}
}
\caption{Proportional Matrices\label{algo:prop}}
\end{algorithm}

Algorithm~\ref{algo:homo} starts with an empty matrix.
Then, it iteratively selects the row (or column) with largest
remaining sum.
Each element of the row (or column) is assigned to the highest average
value.
This avoids large elements in the matrix and leads to low variance.
Algorithm~\ref{algo:hetero} also starts with an empty matrix.
Then, it iteratively assigns the element that can be assigned to the
largest possible value.
This leads to a few large elements in the final matrix.
Algorithm~\ref{algo:prop} starts with the rounding of the rational
proportional matrix (i.e.\ the matrix in which costs are proportional
to the corresponding row and column costs) and proceeds to few random
transformations to meet the constraints.

In Algorithms~\ref{algo:homo} and~\ref{algo:hetero}, the argmin and
argmax can return any index arbitrarily in case of several minimums.
In Algorithm~\ref{algo:hetero}, $\mathbb{1}_n$ denotes a vector of $n$
ones.
Finally, in Algorithms~\ref{algo:hetero} and~\ref{algo:prop}, $\omu^T$
denotes the transpose of $\omu$, which is a column vector.

%%%%%%%%%%%%%%%%%%%%%%%%%%%%%%%%%%%%%%%%%%%%%%%%%%%%%%%%%%%%%%%%%%%%%%%%%%
\subsection{Experiments}\label{sec:xpconv}

We first illustrate the approach with the example of a $20\times10$ matrix with
$N=4\,000$ with given $\omu$ and $\onu$. Starting from three different
matrices as defined in Section~\ref{sec:initial}, we monitor the measures
defined in Section~\ref{sec:measures} in order to observe the
convergence (here, approximately
after 6\,000 iterations). It is, for instance, depicted in Figure~\ref{fig:conv1}
for the cost CV (diagrams for other measures are similar and
seems to converge faster). Next, for every measure, many walks with
different  $\omu$ and $\onu$ (but same $N$) are performed and the value of
the measures is reported in boxplots\footnote{Each boxplot consists
of a bold line for the median, a box for the quartiles, whiskers that
extends to 1.5 times the interquartile range from the box and
additional points for outliers.} for several walking steps, as in
Figure~\ref{fig:conv2} for the CV, allowing to improve the
confidence in  the hypothesis of convergence. One can observe that the three
boxplots are synchronized after about 6\,000 iterations.

\begin{figure*}
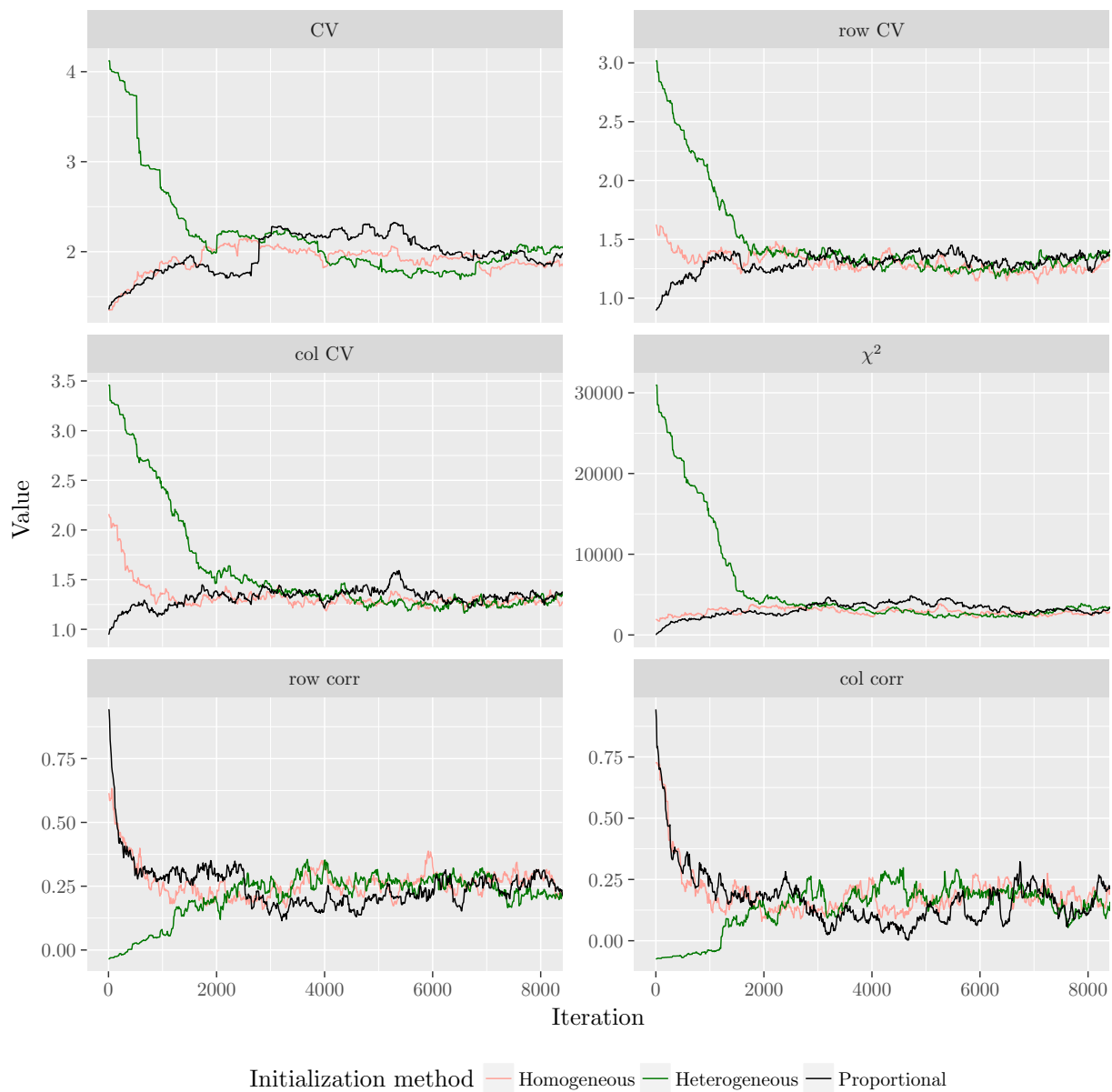

  \small
  \includestandalone[width=\textwidth]{4_convergence_line_RR}
  \caption{Evolution of the measures for a $20\times10$ matrix, with
    $N=20\times n\times m=4\,000$.
    Initial row and column sums ($\omu$ and $\onu$) are generated with
    Algorithm~\ref{algo:nu} without constraints (i.e.\ $\alpha=0$ and
    $\beta=N$).
    Initial matrices are generated with Algorithms~\ref{algo:homo}
    to~\ref{algo:prop} without constraints.}
  \label{fig:conv1}
\end{figure*}

\begin{figure*}
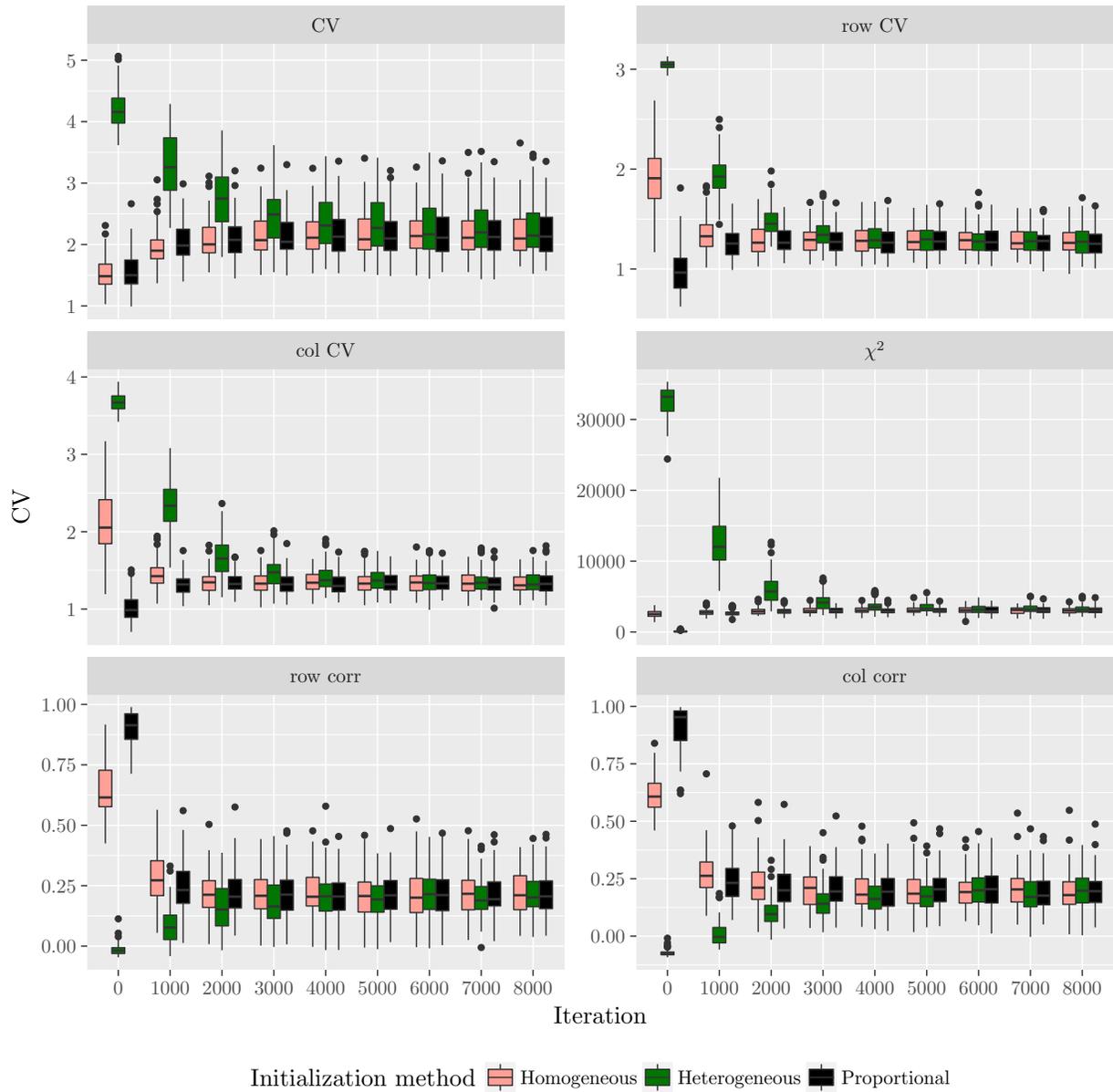

  \small
  \includestandalone[width=\textwidth]{5_convergence_boxplot_RR}
  \caption{Evolution of the measures for matrices with the same
    characteristics as in Figure~\ref{fig:conv1}.
    Each boxplot corresponds to 100 matrices, each based on distinct
    row and column sums.
    When a row or column sum is zero, all undefined measures are
    discarded.}
  \label{fig:conv2}
\end{figure*}

These experiments have been performed for several matrices sizes,
several $\omu$, $\onu$ generations (with different min/max
constraints), and different $N$.
The experimental results seem to point out
that the convergence speed is independent of $N$ (assuming that $N$ is large
enough to avoid bottleneck issues) and independent of the min/max
constraints on $\omu$ and $\onu$. Estimated convergence time (iteration
steps) obtained manually with a visual method (stability for the measures)
for several sizes of matrices are reported in Table~\ref{tab:mixingtimes}.
Experimentally, the mixing (convergence) time seems to be linearly bounded by
$nm \log^3(nm)$.

\begin{table}
\centering
\begin{tabular}{cc}
  \toprule
  $(n,m)$ & mixing time\\
  \midrule
  $(5,5)$ & 200\\
  $(5,10)$ & 600\\
  $(5,15)$ & 1\,000\\
  $(10,10)$ & 2\,500\\
  $(10,15)$ & 3\,500\\
  $(10,20)$ & 6\,000\\
  $(25,10)$ & 7\,500\\
  $(15,20)$ & 8\,000\\
  $(15,25)$ & 13\,000\\
  $(20,25)$ & 30\,000\\
  $(20,30)$ & 50\,000\\
  $(40,20)$ & 65\,000\\
  $(40,40)$ & 210\,000\\
  \bottomrule
\end{tabular}
  \caption{\label{tab:mixingtimes} Estimated mixing times with a
    visual method and with varying number of rows $n$ and columns
    $m$.}
\end{table}

%%%%%%%%%%%%%%%%%%%%%%%%%%%%%%%%%%%%%%%%%%%%%%%%%%%%%%%%%%%%%%%%%%%%%%%%%%
\section{Performance Evaluation of Scheduling Algorithms}
\label{sec:const-effect-cost}

This section studies the effect of the constraints on the matrix
properties (Section~\ref{sec:const-effect-cost-1}) and on the
performance of some scheduling heuristics from the literature
(Section~\ref{section:performance}).

This section relies on matrices of size $20\times10$ with non-zero
cost.
This is achieved by using $\alpha\ge m$ for $\omu$, $\alpha\ge n$ for
$\onu$ and a matrix $\Am$ containing only ones.

Section~\ref{MC:convergence} provides estimation for the convergence
time of the Markov Chain depending on the size of the cost matrix in
the absence of constraints on the vectors ($\alpha$ and $\beta$) and
on the matrix ($\Am$ and $\Bm$).
We assume that the convergence time does not strongly depend on the
constraints.
Moreover, this section relies on an inflated number of iterations,
i.e.\ $50\,000$, for safety, starting from the proportional matrix
(Algorithm~\ref{algo:prop}).

\subsection{Constraints Effect on Cost Matrix Properties}
\label{sec:const-effect-cost-1}

We want to estimate how the constraints on the $\omu$ and $\onu$
random generation influence the matrix properties.
Figure~\ref{fig:conv3} reports the results.
Each row is dedicated to a property from the CV to the column
correlation that are presented in Section~\ref{sec:measures}, with the
inclusion of the $\omu$ and $\onu$ CV.
On the left of the plot, only $\onu$ is constrained.
In the center only $\omu$ and in the right, both $\omu$ and $\onu$.
Constraints are parametrized by a coefficient in
$\lambda \in \{0,0.2,\ldots,1\}$: intuitively, large values of
$\lambda$ impose strong constraints and limit the CV.
The influence of $\lambda$ on the CV of $\omu$ and/or
$\onu$ is consistent with Figure~\ref{fig:variance} in
Section~\ref{section:vectors}: the value decreases from about 20 to 0
as the constraint increases.

\begin{figure}
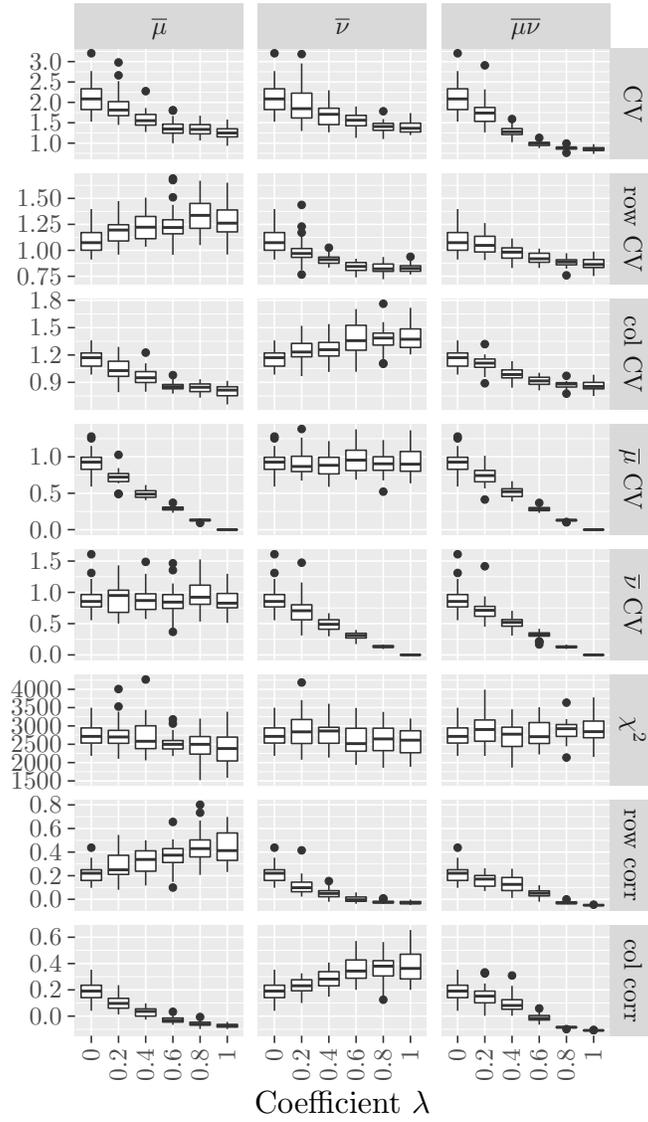

  \centering
  \includestandalone{6_munu_constraints_RR}
  \caption{Values for the measures presented in
    Section~\ref{sec:measures} and the cost sums ($\omu$ and $\onu$)
    CV after $50\,000$ iterations starting with a proportional
    $20\times10$ matrix generated with Algorithm~\ref{algo:prop} with
    different constraints on $\omu$ and/or $\onu$ (either ones on the
    first two columns, both on the third) and with
    $N=20\times n\times m=4\,000$.
    The constraint on $\omu$ (resp.\ $\onu$) is parameterized by a
    coefficient $0\le\lambda\le1$ such that
    $\alpha=\lfloor \frac{\lambda N}{n} \rfloor$ (resp.\
    $\lfloor \frac{\lambda N}{m} \rfloor$) and
    $\beta=\lceil \frac{N}{\lambda n} \rceil$ (resp.\
    $\lceil \frac{N}{\lambda m} \rceil$), with the convention
    $1/0=+\infty$.
    Each matrix contains non-zero costs.
    Each boxplot corresponds to 30 matrices, each based on distinct
    row and column sums.}
  \label{fig:conv3}
\end{figure}

The heterogeneity of a cost matrix can be defined in two
ways\cite{canon2017heterogeneity}: using either the CV of $\omu$ and
$\onu$, or using the mean row and column CV.
Although constraining $\omu$ and $\onu$ limits the former kind of
heterogeneity, the latter only decreases marginally.
To limit the heterogeneity according to both definitions, it is
necessary to constraint the matrix with $\Am$ and $\Bm$.
Figure~\ref{fig:conv4} shows the effect of these additional
constraints when the cost matrix cannot deviate too much from an ideal
fractional proportional matrix.
In particular, $\omu$ (resp.\ $\onu$) is constrained with a parameter
$\lambda_r$ (resp.\ $\lambda_c$) as before.
The constraint on the matrix is performed with the maximum $\lambda$
of these two parameters.
This idea is to ensure the matrix is similar to a proportional matrix
$M$ with $M(i,j)=\frac{\omu(i)\times\onu(j)}{N}$ when any constraint
on the row or column sum vectors is large.

\begin{figure*}
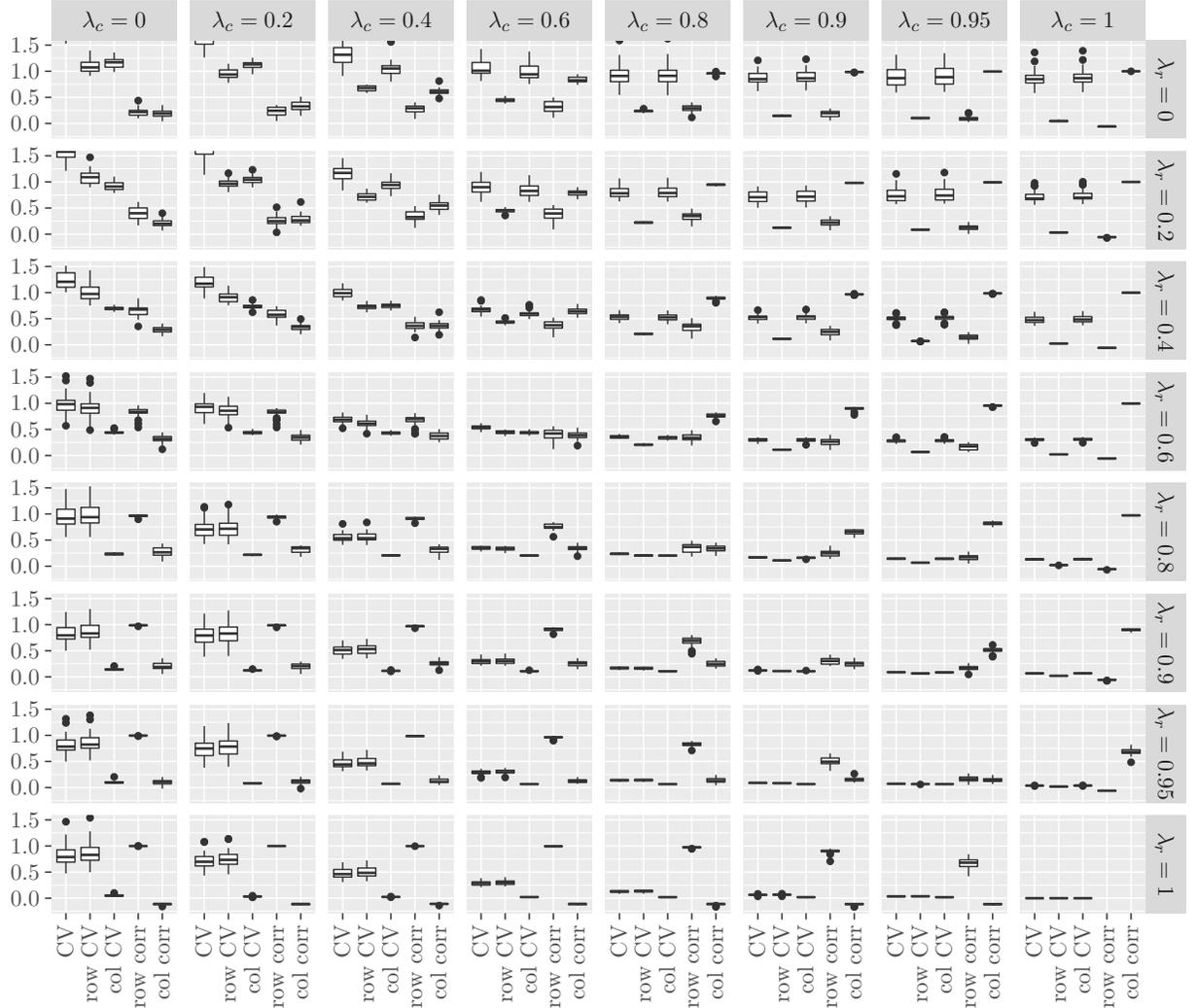

  \includestandalone[width=\textwidth]{7_AminBmax_constraints_RR}
  \caption{Values for the measures presented in
    Section~\ref{sec:measures} and the cost sums ($\omu$ and $\onu$)
    CV after $50\,000$ iterations starting with a proportional
    $20\times10$ matrix generated with Algorithm~\ref{algo:prop} with
    different constraints on $\omu$, $\onu$ and the matrix, and with
    $N=20\times n\times m=4\,000$.
    The constraint on $\omu$ (resp.\ $\onu$) is parameterized by a
    coefficient $0\le\lambda_r\le1$ (resp.\ $0\le\lambda_c\le1$) such
    that $\alpha=\lfloor \frac{\lambda_r N}{n} \rfloor$ (resp.\
    $\lfloor \frac{\lambda_c N}{m} \rfloor$) and
    $\beta=\lceil \frac{N}{\lambda_r n} \rceil$ (resp.\
    $\lceil \frac{N}{\lambda_c m} \rceil$), with the convention
    $1/0=+\infty$.
    The constraint on the matrix is parameterized by a coefficient
    $\lambda=\max(\lambda_r,\lambda_c)$ such that
    $\Am=\lfloor \lambda M \rfloor$ and
    $\Bm=\lceil M / \lambda \rceil$ with
    $M(i,j)=\frac{\omu(i)\times\onu(j)}{N}$.
    Each matrix contains non-zero costs.
    Each boxplot corresponds to 30 matrices, each based on distinct
    row and column sums.
    When $\lambda_r=\lambda_c=1$, all costs are identical and the
    correlations are discarded.}
  \label{fig:conv4}
\end{figure*}

The figure shows that the cost CV decreases as both $\lambda_r$ and
$\lambda_c$ increase.
Moreover, as for the $\omu$ (resp.\ $\onu$) CV, the mean column
(resp.\ row) CV decreases as $\lambda_r$ (resp.\ $\lambda_c$)
increases.
We can thus control the row and column heterogeneity with
$\lambda_r$ and $\lambda_c$, respectively.
Note that when reducing the heterogeneity, row or column correlations
tend to increase.
In particular, large values for $\lambda_r$/$\lambda_c$ lead to jumps
from small correlations when $\lambda_r=\lambda_c$ to large row
(resp.\ column) correlation when $\lambda_r=1$ (resp.\ $\lambda_c=1$).

%%%%%%%%%%%%%%%%%%%%%%%%%%%%%%%%%%%%%%%%%%%%%%%%%%%%%%%%%%%%%%%%%%%%%%%%%%
\subsection{Constraints Effect on Scheduling Algorithms}\label{section:performance}

Generating random matrices with parameterized constraints allows the
assessment of existing scheduling algorithms in different contexts.
In this section, we focus on the impact of cost matrix properties on
the performance of three heuristics for the problem denoted
$R||C_{max}$.
This problem consists in assigning a set of independent tasks to
machines such that the \emph{makespan} (i.e.\ maximum completion time
on any machine) is minimized.
The cost of any task on any machine is provided by the cost matrix and
the completion time on any machine is the sum of the costs of all task
assigned to it.

The heuristics we consider constitute a diversified selection among
the numerous heuristics that have been proposed for this problem in
terms of principle and cost:
\begin{itemize}
\item \emph{BalSuff}, an efficient
  heuristic\cite{canon2017heterogeneity} with unknown complexity that
  balances each task to minimize the makespan.
\item \emph{HLPT}, Heterogeneous-Longest-Processing-Time, iteratively
  assigns the longest task to the machine with minimum completion time
  in $O(nm+n\log(n))$ steps.
  This is a natural extension of LPT\cite{graham69} and variant of
  HEFT\cite{topcuoglu2002performance} in which the considered cost for
  each task is its minimal one.
\item \emph{EFT}, Earliest-Finish-Time, (or MinMin) is a classic
  principle, which iteratively assigns each task by selecting the task
  that finishes the earliest on any machine.
  Its time complexity is $O(n^2m)$.
\end{itemize}

We selected four scenarios that represent the extremes in terms of
parameters, heterogeneity and correlation: $\lambda_r=\lambda_c=0$
with the most heterogeneity and the least correlation,
$\lambda_r=0,\lambda_c=1$ with a high task and low machine
heterogeneity, $\lambda_r=1,\lambda_c=0$ with a low task and high
machine heterogeneity, and $\lambda_r=0.75,\lambda_c=1$ with low
heterogeneity and high correlation (the case $\lambda_r=\lambda_c=1$
lead to identical costs for which all heuristics perform the same).
Table~\ref{tab:conv5} gives the mean properties for each scenario with
100 matrices each.
Figure~\ref{fig:conv5} depicts the results: for each scenario and
matrix, the makespan for each heuristic was divided by the best one
among the three.
All heuristics exhibit different behaviors that depends on the
scenario.
BalSuff outperforms its competitors except when $\lambda_r=0.75$ and
$\lambda_c=1$, in which case it is even the worst.
HLPT is always the best when $\lambda_c=1$.
In this case, each task has similar costs on any machine.
This corresponds to the problem $P||C_{\max}$, for which LPT, the
algorithm from which is inspired HLPT, was proposed with an
approximation ratio of 4/3\cite{graham69}.
The near-optimality of HLPT for instances with large row and low
column heterogeneity is consistent with the
literature\cite{canon2017heterogeneity}.
Finally, EFT performs poorly except when $\lambda_r=1$ and
$\lambda_c=0$.
In this case, tasks are identical and it relates to the problem
$Q|p_i=1|C_{\max}$.
These instances, for which the row correlation is high and column
correlation is low, have been shown to be the easiest for
EFT\cite{canon2017controlling}.

\begin{figure}
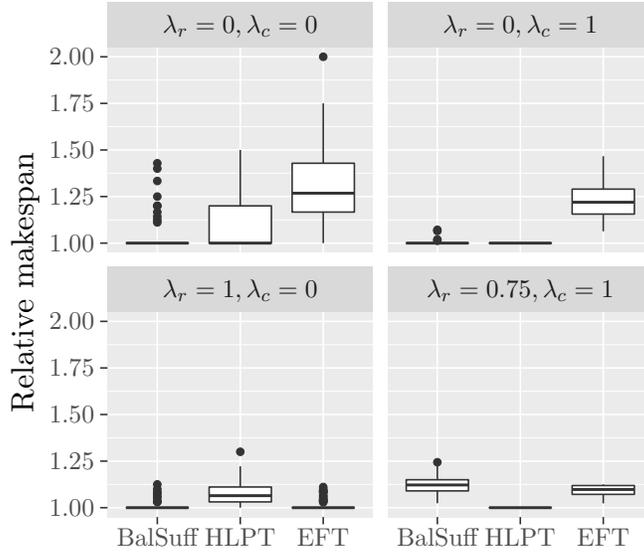

  \centering
  \includestandalone{8_heuristics}
  \caption{Ratios of makespan to the best among BalSuff, HLPT and EFT.
    The cost matrices were generated as in Figure~\ref{fig:conv4}.
    Each boxplot corresponds to 100 cost matrices, each based on
    distinct row and column sums.
  }
  \label{fig:conv5}
\end{figure}

\begin{table*}
  \centering
  \begin{tabular}{cccccccccc}
    \toprule
    $\lambda_r$ & $\lambda_c$ & CV & row CV & col CV & $\omu$ CV & $\onu$ CV & $\chi^2$ & row corr & col corr\\
    \midrule
    0 & 0 & 2.1 & 1.1 & 1.2 & 0.9 & 0.87 & 2831 & 0.21 & 0.18\\
    0 & 1 & 0.88 & 0.051 & 0.9 & 0.9 & 0 & 4.2 & -0.058 & 1\\
    1 & 0 & 0.86 & 0.9 & 0.051 & 0 & 0.9 & 1.7 & 1 & -0.11\\
    0.75 & 1 & 0.17 & 0.02 & 0.17 & 0.17 & 0 & 4.3 & -0.058 & 0.98\\
    \bottomrule
  \end{tabular}
  \caption{\label{tab:conv5} Mean properties over 100 matrices
    generated as in Figure~\ref{fig:conv5} for each pair of parameters
    $\lambda_r$ and $\lambda_c$.}
\end{table*}

%%%%%%%%%%%%%%%%%%%%%%%%%%%%%%%%%%%%%%%%%%%%%%%%%%%%%%%%%%%%%%%%%%%%%%%%%%
\section{Conclusion}

Random instance generation allows broader experimental campaigns but
can be hindered by bias in the absence of guarantee on the
distribution of the instances.
This work focuses on the generation of cost matrices, which can be
used in a wide range of scheduling problems to assess the performance
of any proposed solution.
We propose a Markov Chain Monte Carlo approach to draw random cost
matrices from a uniform distribution: at each iteration, some costs in
the matrix are shuffled such that the sum of the costs on each row and
column remains unchanged.
By proving its ergodicity and symmetry, we ensure that its stationary
distribution is uniform over the set of feasible instances.
Moreover, the result holds when restricting the set of feasible
instances to limit their heterogeneity.
Finally, experiments were consistent with previous studies in the
literature.
Although constraining the matrix generation with a minimum and maximum
matrices leads to large correlations, it remains to
determine the drawbacks of this approach and whether there could be
more relevant solutions (such as using a simulated annealing to
increase the correlations starting from a uncorrelated ones).
A more prospective future direction would be to apply the current
methodology on the generation of other types of instances such as task
graphs.

\section*{Acknowledgments}

The authors would like to thank Anne Bouillard for pointing out works on
contingency tables.

\bibliographystyle{IEEEtran}
\bibliography{biblio_RR}

\appendix

\section{Notation}
\label{sec:notation}

Table~\ref{tab:notations} provides a list of the most used notations in this
report.

\begin{table}[h]
\centering
\begin{tabular}{rm{0.6\columnwidth}}
\toprule
Symbol & Definition\tabularnewline
\midrule
$n$ & Number of rows (tasks)\tabularnewline
$m$ & Number of columns (machines)\tabularnewline
$M(i,j)$ & Element on the $i$th row and $j$th column of matrix
$M$\tabularnewline
$N$ & Sum of elements in a matrix ($\sum_{i,j}M(i,j)$)\tabularnewline
$\omu$ & Vector of size $n$. $\frac{\omu(i)}{m}$ is the mean cost
of the $i$th task. \tabularnewline
$\onu$ & Vector of size $m$. $\frac{\onu(j)}{n}$ is the mean cost
on the $j$th machine. \tabularnewline
\midrule
$H_{N,n}^{\alpha,\beta}$ & Elements $\overline{v}\in \mathbb{N}^n$ s.t.
$\alpha\leq \overline{v}(i)\leq \beta$ and $\sum_{i=1}^n
\overline{v}(i)=N$.\tabularnewline
$h_{N,n}^{\alpha,\beta}$ & Cardinal of $H_{N,n}^{\alpha,\beta}$.\tabularnewline
\midrule
$d(A,B)$ & Distance between matrices $A$ and $B$.\tabularnewline
$\Om$ & Set of contingency tables of sum $N$ and sums of rows and columns
$\omu$ and $\onu$.\tabularnewline
$\alpha$, $\beta$ & Scalar constraints on minimal/maximal values for
generated matrices.\tabularnewline
$\oalpha$, $\obeta$ & Vector constraints on
minimal/maximal values for generated matrices.\tabularnewline
$\Am$, $\Bm$ & Matrix constraints on minimal/maximal values for
generated matrices.\tabularnewline
$\Om[...]$ & Subset of $\Om$ min/max-constrained by $[...]$.\tabularnewline
$f(\cdot,\cdot)$ & Random mapping for the Markov Chains.\tabularnewline
\bottomrule
\end{tabular}
\caption{\label{tab:notations}List of notations.}
\end{table}

\end{document}